\documentclass[a4paper,UKenglish]{article}
 
\usepackage{microtype}

\title{To Infinity and Beyond\footnote{This work was partially supported by the French \emph{Agence Nationale pour la~Recherche},
through the Project {\bf MealyM} ANR-JS02-012-01.
}
}

\author{Ines Klimann\\
Univ Paris Diderot, Sorbonne Paris Cit\'e, IRIF,\\
  UMR 8243 CNRS, F-75013 Paris, France,\\
  \texttt{klimann@irif.fr}
}

\usepackage{url}
\usepackage{amsmath,amsfonts,amsthm,amssymb}
\usepackage{amsopn}
\usepackage{pgf,tikz}
\usetikzlibrary{arrows,automata,shapes}
\tikzstyle{every state}=[minimum size=12pt,inner sep=0pt]
\usepackage{xspace}
\usepackage{stmaryrd}
\usepackage{wasysym}
\usepackage{textcomp}
\usepackage[margin=1in]{geometry}

\newcommand{\ie}{\emph{i.e.}\xspace}
\newcommand{\wrt}{w.r.t. }
\newcommand{\resp}{\emph{resp.}\xspace}
\newcommand{\eg}{\emph{e.g.}\xspace}
\newcommand{\aut}[1]{{\mathcal #1}}
\newcommand{\mot}[1]{{\mathbf {#1}}}
\newcommand{\pres}[1]{\langle{#1}\rangle}
\newcommand{\presm}[1]{\pres{{#1}}_{+}}
\DeclareMathOperator{\mz}{\mathfrak m}
\newcommand{\lacroix}{\tikz[baseline=-.5ex]{\draw[->,>=latex] (-.2ex,0) -- (3.8ex,0); \draw[->,>=latex] (1.7ex,2ex) -- (1.7ex,-2.3ex);}}

\newcommand{\N}{{\mathbb N}}

\DeclareMathOperator{\CC}{cc}
\newcommand{\cc}[1]{\CC({#1})}
\newcommand{\ccNerode}[1]{\llbracket{#1}\rrbracket}
\DeclareMathOperator{\llN}{\llbracket}
\DeclareMathOperator{\rrN}{\rrbracket\hspace*{-1.1ex}\bullet}
\newcommand{\Nerendq}[2][q]{\llN{{#1}^{#2}}\rrN\xspace}
\newcommand{\follow}[2]{\{{#1}?\looparrowright{#2}\}}
\newcommand{\precede}[2]{\{?{#1}\looparrowright{#2}\}}
\newcommand{\lettre}[2]{{#2}\text{\rm \textlquill}{#1}\text{\rm \textrquill}}

\theoremstyle{plain}
\newtheorem{theorem}{Theorem}
\newtheorem{proposition}[theorem]{Proposition}
\newtheorem{lemma}[theorem]{Lemma}
\newtheorem{corollary}[theorem]{Corollary}

\begin{document}

\maketitle

\begin{abstract}
We prove that if a group generated by a bireversible Mealy automaton
contains an element of infinite order, its growth blows up and is necessarily
exponential. As a direct consequence, no infinite virtually nilpotent group can be generated by
a bireversible Mealy automaton.
\end{abstract}

\section*{}

The study on how (semi)groups grow has been highlighted since Milnor's
question on the existence of groups of intermediate growth (faster
than any polynomial and slower than any exponential) in
1968~\cite{milnor}, and the very first example of such a group given
by Grigorchuk~\cite{grigorch:degrees}. Uncountably many examples have
followed this first one, see for
instance~\cite{Grigorchuk85}. Bartholdi and Erschler have even
obtained results on precise computations of growth, in particular they
proved that if a function satisfies some frame property, then there
exists a finitely generated group with growth equivalent to
it~\cite{BE14}. Besides, for now, intermediate growth and automaton
groups, that is groups generated by Mealy automata, seem to have a
very strong link, since the only known examples of intermediate growth
groups are either automaton groups, or based on such groups.

There exists no criterium to test if a Mealy automaton generates a group of
intermediate growth and it is not even known if this property is
decidable. However, there is no known example in the litterature of a
bireversible Mealy automaton generating an intermediate growth group
and it is legitimate to wonder if it is possible. This
article enter in this scope. We prove that if there exists at least
one element of infinite order in a group generated by a bireversible
Mealy automaton, then its growth is necessarily exponential. It has
been conjuctered, and proved in some cases~\cite{GK17}, that an infinite group
generated by a bireversible Mealy automaton always has an element of
infinite order, which suggests that, indeed, a group
generated by a bireversible Mealy automaton either is finite, or has
exponential growth.

Finally, let us mention the work by Brough and Cain to obtain some
criteria to decide if a semigroup is an automaton semigroup~\cite{BC17}. Our
work can be seen as partially answering a similar question: can a
given group be generated by a bireversible Mealy automaton? A
consequence of our result is that no infinite virtually nilpotent group can be.

This article is organized as follows. In Section~\ref{sec-basic}, we
define the automaton groups and the growth of a group,
and give some properties on the connected components of the powers
of a Mealy automaton. In Section~\ref{sec-equiv}, we study the
behaviour of some equivalence classes of words on the state set of a
Mealy automaton. Finally, the main result takes place in Section~\ref{sec-main}.

\section{Basic notions}\label{sec-basic}

In all the article, if \(E\) is a finite set, its cardinality is
denoted by \(|E|\). A \emph{finite word\/} of \emph{length\/}~\(n\) 
on~\(E\) is a finite sequence of~\(n\) elements of~\(E\) and is
denoted classically as the concatenation of its elements.
The set of finite words over~\(E\) is denoted by~\(E^*\), the
set of non-empty finite words by~\(E^+\), and the set of words of
length~\(n\) by~\(E^n\). In general the elements
of~\(E\) are written in plain letters, \eg~\(q\), while the
words on~\(E\) are written in bold letters, \eg \(\mot{u}\). The
length of~\(\mot{u}\) is denoted by~\(|\mot{u}|\), its letters are
numbered from~\(0\) to~\(\mot{u}-1\) and if \(i\) is an integer, its
\((i\mod|\mot{u}|)\)-th letter is denoted by \(\lettre{i}{\mot{u}}\);
for example its first letter is~\(\lettre{0}{\mot{u}}\), while its
last letter is \(\lettre{-1}{\mot{u}}\). If
\(L\) is a set of words on~\(E\), \(\lettre{i}{L}\) denotes the set
\(\{\lettre{i}{\mot{u}} \mid \mot{u}\in L\}\).

\subsection{Semigroups and groups generated by Mealy automata}
We first recall the formal definition of an automaton. A {\em (finite, deterministic, and complete) automaton} is a
triple
\(
\bigl( Q,\Sigma,\delta = (\delta_i\colon Q\rightarrow Q )_{i\in \Sigma} \bigr)
\),
where the \emph{state set\/}~$Q$
and the \emph{alphabet\/}~$\Sigma$ are non-empty finite sets, and
the~\(\delta_i\) are functions.

A \emph{Mealy automaton\/} is a quadruple
\(\aut{A}=( Q, \Sigma, \delta,\rho)\),
such that \((Q,\Sigma,\delta)\) and~\((\Sigma,Q,\rho)\) are both
automata.
In other terms, a Mealy automaton is a complete, deterministic,
letter-to-letter transducer with the same input and output
alphabet. Its \emph{size\/} is the cardinality of its state set and is
denoted by \(\#\aut{A}\).

The graphical representation of a Mealy automaton is
standard, see Figure~\ref{fig-aleshin}. But, for practical reasons, we use sometimes other
graphical representations for the transitions. For example the
transition from~\(x\) to~\(z\) with input letter~\(0\) and output
letter~\(1\) in the automaton of Figure~\ref{fig-aleshin} can be
represented
\[\text{either by}\quad x\xrightarrow{0|1}z\:,\qquad\text{or by}\quad\begin{array}{ccc}
& 0 &\\
x & \lacroix & z\\
& 1 &\\
\end{array}\:.\]

\begin{figure}[h]
\centering
\begin{tikzpicture}[->,>=latex,node distance=1.8cm]
\tikzstyle{every state}=[minimum size=12pt,inner sep=0pt]
\node[state] (z) {\(z\)};
\node[state] (y) [below left of=z] {\(y\)};
\node[state] (x) [above left of=y] {\(x\)};
\node (so) [left of=y] {};
\node (ne) [right of=y] {};
\path (x) edge node[below,near start]{\(1|0\)} (y)
      (x) edge [bend left] node[above]{\(0|1\)} (z)
      (z) edge [bend left] node[above]{\(0|0\), \(1|1\)} (x)
      (y) edge node[below,near end]{\(1|0\)} (z)
      (y) edge [loop right] node{\(0|1\)} (y);
\end{tikzpicture}
\caption{The Aleshin automaton.}\label{fig-aleshin}
\end{figure}
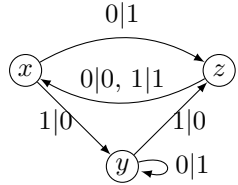

Let~\(\aut{A} = (Q,\Sigma, \delta,\rho)\) be a Mealy automaton.
Each state \(q\in Q\) defines a mapping from \(\Sigma^*\) into itself,
recursively by:
\begin{equation*}
\forall i \in \Sigma, \ \forall \mot{s} \in \Sigma^*, \qquad
\rho_q(i\mot{s}) = \rho_q(i)\rho_{\delta_i(q)}(\mot{s}) \:.
\end{equation*}

The image of the empty word is itself. For each $q\in Q$,
the mapping~\(\rho_q\) is length-preserving and prefix-preserving.
We say that~\(\rho_q\) is the function \emph{induced\/} by \(q\).
For~$\mot{u}=q_1\cdots q_n \in Q^n$ with~$n>0$, set
\(\rho_\mot{u}\colon\Sigma^* \rightarrow \Sigma^*, \rho_\mot{u} = \rho_{q_n}
\circ \cdots \circ \rho_{q_1} \:\).

The semigroup of mappings from~$\Sigma^*$ to~$\Sigma^*$ generated by
$\{\rho_q, q\in Q\}$ is called the \emph{semigroup generated
  by~$\aut{A}$\/} and is denoted by~$\presm{\aut{A}}$.

\medskip

A Mealy automaton \(\aut{A}=(Q,\Sigma,\delta, \rho)\) is
\emph{invertible\/} if the functions \(\rho_q\) are permutations
of the alphabet~\(\Sigma\). In this case, the functions induced by the states are
permutations on words of the same length and thus we may consider
the group of mappings from~$\Sigma^*$ to~$\Sigma^*$ generated by
$\{\rho_q, q\in Q\}$: it is called the \emph{group generated
  by~$\aut{A}$\/} and is denoted by~$\pres{\aut{A}}$.

When \(\aut{A}\) is invertible, define its \emph{inverse\/}
\(\aut{A}^{-1}\) as the Mealy automaton with state set~\(Q^{-1}\), a
disjoint copy of~\(Q\), and alphabet~\(\Sigma\), where the transition \(p^{-1}
\xrightarrow{j\mid i} q^{-1}\) belongs to~\(\aut{A}^{-1}\) if and only
if the transition~\(p \xrightarrow{i\mid j} q\) belongs
to~\(\aut{A}\). Clearly the action induced by the state~\(p^{-1}\) of
\(\aut{A}^{-1}\) is the reciprocal of the action induced by the
corresponding state~\(p\) in~\(\aut{A}\).

\smallskip

A Mealy automaton \((Q,\Sigma,\delta, \rho)\) is
\emph{reversible\/} if the functions \(\delta_i\) induced on~\(Q\) by
the input letters of the transitions are
permutations. The connected components of a reversible
automaton are strongly connected.
In a reversible automaton of state set~\(Q\) and
alphabet~\(\Sigma\), for any word~\(\mot{s}\in\Sigma^*\) and any
state~\(q\), there exists exactly one path in the automaton with
label~\(\mot{s}\) and final state~\(q\), hence we can consider the
\emph{backtrack application\/} induced by~\(q\): it associates to \(\mot{s}\)
the output label~\(\mot{t}\in\Sigma^{|\mot{s}|}\) of this single path.

A Mealy automaton is \emph{coreversible\/} if the functions induced
on~\(Q\) by the letters as output letters of the transitions are
permutations.

A Mealy automaton is \emph{bireversible\/} if it is both reversible
and coreversible. It is quite simple to see that the applications
and the backtrack applications induced by the states of a bireversible
automaton are permutations.

\medskip

Two Mealy automata are said to be \emph{isomorphic\/} if they are
identical up to the labels of their states.

\medskip

We extend to \(\delta\) the former notations on
\(\rho\), in a natural way. Hence \(\delta_i\colon Q^*\rightarrow Q^*,
i\in \Sigma\), are the functions extended to~\(Q^*\), and
for~$\mot{s}=i_1\cdots i_n \in \Sigma^n$ with~$n>0$, we
set~\(\delta_\mot{s}\colon Q^* \rightarrow Q^*, \ \delta_\mot{s} =
\delta_{i_n}\circ \cdots \circ \delta_{i_1}\).


\subsection{Growth of a semigroup or of a group}
Let \(H\) be a semigroup generated by a finite set \(S\). The
\emph{length\/} of an element \(g\) of the 
semigroup, denoted by \(|g|\), is the length of its shortest
decomposition as a product of generators: \[|g| = 
\min\{ n\mid \exists s_1,\dots,s_n\in S,\, g=s_1\cdots
s_n\}\:.\]

The \emph{growth function\/} \(\gamma_H^S\) of the semigroup \(H\) with
respect to the generating set~\(S\) enumerates the elements of~\(H\)
with respect to their length:
\[\gamma_H^S(n) = |\{g\in H\,;\, |g|\leq n\}|\:.\]
The \emph{growth functions\/} of a group are defined similarly by taking
symmetrical generating sets.

\medskip

The growth functions corresponding to two generating sets are
equivalent~\cite{Mann}, so we may define the \emph{growth\/} of a group
or a semigroup as the equivalence class of its growth
functions. Hence, for example, a finite (semi)group has a bounded 
growth, while an infinite abelian (semi)group has a polynomial growth,
and a non-abelian free (semi)group has an exponential growth.

\bigskip

It is quite easy to obtain groups of polynomial or exponential
growth. Answering a question of Milnor~\cite{milnor}, Grigorchuk gave
an example of an automaton group of intermediate
growth~\cite{grigorch:degrees}: faster than any polynomial, slower than
any exponential, opening thus a new classification criterium for
groups, that has been deeply studied since this seminal
article (see~\cite{GP06} and references therein). Besides,
intermediate growth and automaton groups 
seem to have a very strong link, since the only known examples of
intermediate growth groups in the literature are based on automaton groups.

\medskip

Note an important point for our purpose: let \(G\) be a group
finitely generated by~\(S\), and \((I_n)_{n>0}\) a sequence of subsets
of \(G\), compatible with the length of the elements, \ie the sets~\(I_n\)
are pairwise distinct and the elements of~\(I_n\) have all length less
than or equal to~\(n\). The growth function of \((I_n)_{n>0}\) is given by
\((\sum_{k\leq n}|I_n|)_{n>0}\); if it grows exponentially, then so does~\(G\).
In the same spirit, a group which admits a subgroup of exponential
growth grows exponentially.

\subsection{The powers of a Mealy automaton and their connected components}\label{sec-powers}
The powers of a Mealy automaton have been shown to play an important
role in the finiteness and the order problem for an automaton (semi)groups, as
highlighted in~\cite{Kli13,klimann_ps:3state,GK17}. The \emph{\(n\)-th
  power\/} of the automaton 
\(\aut{A}=(Q,\Sigma,\delta,\rho)\) is the Mealy automaton
\begin{equation*}
\aut{A}^n = \bigl( \ Q^n,\Sigma, (\delta_i\colon Q^n \rightarrow
Q^n)_{i\in \Sigma}, (\rho_{\mot{u}}\colon \Sigma \rightarrow \Sigma
)_{\mot{u}\in Q^n} \ \bigr)\enspace.
\end{equation*}

Note that the powers of a reversible (\resp bireversible) Mealy
automaton are reversible (\resp bireversible).

The (semi)group generated by a connected component of some power
of~\(\aut{A}\) is a sub(semi)group of the (semi)group generated
by~\(\aut{A}\).

Let \(\mot{u}\) and \(\mot{v}\) be elements of \(Q^+\) and \(\aut{C}\)
be a connected component of some power of \(\aut{A}\): \(\mot{v}\)
\emph{can follow \(\mot{u}\) in \(\aut{C}\)\/} if \(\mot{uv}\) is the 
prefix of some state of \(\aut{C}\). We denote by
\(\follow{\mot{u}}{\aut{C}}\) the set of the states which
can follow \(\mot{u}\) in~\(\aut{C}\):
\[\follow{\mot{u}}{\aut{C}} = \{q\in Q\mid \mot{u}q\text{ is the
  prefix of some state of }\aut{C}\}\:.\]

We define similarly the fact that \(\mot{v}\) \emph{can precede\/}
\(\mot{u}\) in \(\aut{C}\) if \(\mot{vu}\) is the 
suffix of some state of \(\aut{C}\), and we introduce the set
\[\precede{\mot{u}}{\aut{C}} = \{q\in Q\mid q\mot{u}\text{ is the
  suffix of some state of }\aut{C}\}\:.\]

The aim of this section is to give some intuition on the links between
the connected components of consecutive powers of~\(\aut{A}\). Since a
word can be extended with a prefix or a suffix, most of the results
exposed here are expressed in both cases, but only the first result is
proved in both cases, to show how bireversibility allow to consider
similarly the actions and the backtrack actions.

\begin{lemma}\label{lem-setF}
Let \(\aut{A}\) be a bireversible Mealy automaton with state
set~\(Q\) and \(\aut{C}\) a connected component of one of its powers. If
\(\mot{u}\in Q^+\) is a proper prefix of some state of~\(\aut{C}\), then the cardinality
of the set \(\follow{\mot{u}}{\aut{C}}\) depends only on the length of
\(\mot{u}\).
\end{lemma}

\begin{proof}
Suppose that \(\mot{u'}\) is such that \(\mot{uu'}\) is a state
of~\(\aut{C}\), and let \(\mot{v}\) be a prefix of some
state~\(\mot{vv'}\) of~\(\aut{C}\) with the same length
as~\(\mot{u}\). Since \(\aut{C}\) is a connected component in a
reversible Mealy automaton, it is strongly connected, so there exists
a word \(\mot{s}\in\Sigma^*\) such that \(\delta_{\mot{s}}(\mot{uu'})=
\mot{vv'}\). Now, consider the action induced by~\(\mot{s}\) on
\(\mot{u}p\), for \(p\in\follow{\mot{u}}{\aut{C}}\):
\[\begin{array}{ccc}
& \mot{s} &\\
\mot{u} & \lacroix & \mot{v}\\
& \mot{s'} &\\
p & \lacroix & p'\\
\end{array}\]
Since the automaton~\(\aut{A}\) is reversible, the action induced by~\(\mot{s'}\)
is a permutation of~\(Q\), and we have
\[|\follow{\mot{u}}{\aut{C}}|\leq |\follow{\mot{v}}{\aut{C}}|\:.\] The
reciprocal inequality is obtained symmetrically.
\end{proof}

\begin{lemma}
Let \(\aut{A}\) be a bireversible Mealy automaton with state
set~\(Q\) and \(\aut{C}\) a connected component of one of its powers. If
\(\mot{u}\in Q^+\) is a proper suffix of some state of~\(\aut{C}\), then the cardinality
of the set \(\precede{\mot{u}}{\aut{C}}\) depends only on the length
of \(\mot{u}\).
\end{lemma}

\begin{proof}
Suppose that \(\mot{u'}\) is such that \(\mot{u'u}\) is a state
of~\(\aut{C}\), and let \(\mot{v}\) be a suffix of some
state~\(\mot{v'v}\) of~\(\aut{C}\) with the same length
as~\(\mot{u}\). Since \(\aut{C}\) is a connected component in a
reversible Mealy automaton, it is strongly connected, so there exist
words \(\mot{s},\mot{t}\in\Sigma^*\) such that \(\delta_{\mot{s}}(\mot{u'u})=
\mot{v'v}\) and \(\rho_{\mot{u'u}}(\mot{s}) = \mot{t}\):
\[\begin{array}{ccc}
& \mot{s} &\\
\mot{u'} & \lacroix & \mot{v'}\\
& \mot{s'} &\\
\mot{u} & \lacroix & \mot{v}\\
& \mot{t} &
\end{array}\]
Now, consider the backtrack action induced by~\(\mot{t}\) on
\(p\mot{u}\), for \(p\in\precede{\mot{u}}{\aut{C}}\):
\[\begin{array}{ccc}
p & \lacroix & p'\\
& \mot{s'} &\\
\mot{u} & \lacroix & \mot{v}\\
& \mot{t} &
\end{array}\]
Since the automaton~\(\aut{A}\) is bireversible, the backtrack action
induced by~\(\mot{s'}\) is a permutation of~\(Q\), and we have
\[|\precede{\mot{u}}{\aut{C}}|\leq |\precede{\mot{v}}{\aut{C}}|\:.\] The
reciprocal inequality is obtained symmetrically.
\end{proof}

Consider a state \(q\) of \(\aut{A}\). For any integer \(n>0\), we
denote by \(\cc{q^n}\) the connected component of \(q^n\) in
\(\aut{A}^n\). The sequence of such components has some
properties which we give here. These properties can be seen as
properties of the branch represented by~\(q^{\omega}\) in the Schreier
trie of~\(\aut{A}\) (also known as the orbit tree of the dual of~\(\aut{A}\))
which has been introduced in~\cite{klimann_ps:3state,GK17} (to keep
this article self-contained, we give here only the properties of this
branch, but for a more global intuition on the constructions, the
reader can  consult these references).

The first point is given by Lemma~\ref{lem-setF}: for any \(n>0\),
the component 
\(\cc{q^{n+1}}\) can be seen as several full copies of the component
\(\cc{q^n}\); indeed, if \(\mot{u}\) and \(\mot{v}\) are states of
\(\cc{q^n}\), then in \(\cc{q^{n+1}}\) there are as many states  with
prefix~\(\mot{u}\) as states with prefix~\(\mot{v}\), \ie
\[\forall\mot{u}\in\cc{q^n},\forall\mot{v}\in\cc{q^n},\,
|\follow{\mot{u}}{\cc{q^{n+1}}}| =
|\follow{\mot{v}}{\cc{q^{n+1}}}|\:.\]
Hence the ratio between the size of~\(\cc{q^{n+1}}\) and the size
of~\(\cc{q^n}\) is necessarily an integer: it is the cardinality
of the set \(\follow{\mot{u}}{\cc{q^{n+1}}}\) for any
state~\(\mot{u}\) of~\(\cc{q^n}\), and in particular
for~\(\mot{u}=q^n\).

We define by \[\left(\frac{\#\cc{q^{n+1}}}{\#\cc{q^n}}\right)_{n>0}=
\left(|\follow{q^n}{\cc{q^{n+1}}}|\right)_{n>0}\]
the \emph{sequence of ratios\/} associated to the state \(q\).

\begin{lemma}\label{lem-follow}
  Let \(\aut{A}\) be a bireversible Mealy automaton with state
set~\(Q\), \(q\in Q\) a state of~\(\aut{A}\), and \(\mot{u}\) and
\(\mot{v}\) be two elements of~\(Q^+\). If \(\mot{uv}\) is a state of
\(\cc{q^{|\mot{uv}|}}\), then \(\mot{v}\) is a state of
\(\cc{q^{|\mot{v}|}}\)
and \[\follow{\mot{uv}}{\cc{q^{|\mot{uv}|+1}}}\subseteq
\follow{\mot{v}}{\cc{q^{|\mot{v}|+1}}}\:.\]
\end{lemma}

\begin{proof}
Take \(\aut{A}=(Q,\Sigma,\delta,\rho)\) and let \(p\in
\follow{\mot{uv}}{\cc{q^{|\mot{uv}|+1}}}\).
  
Since \(\aut{A}\) is reversible, there exist a
word~\(\mot{s}\in\Sigma^*\) such
that~\(\delta_{\mot{s}}(\mot{uv}p)=q^{|\mot{uv}|+1}\):
\[\begin{array}{ccc}
& \mot{s} &\\
\mot{u} & \lacroix & q^{|\mot{u}|}\\
& \mot{s'} &\\
\mot{v} & \lacroix & q^{|\mot{v}|}\\
& \mot{s''} &\\
p & \lacroix & q\\
\end{array}\]
Hence \(\delta_{\mot{s'}}(\mot{v})=q^{\mot{v}}\) and so
\(\mot{v}\in\cc{q^{|\mot{v}|}}\), and
\(\delta_{\mot{s'}}(\mot{v}p)=q^{|\mot{v}|+1}\), which means that
\(p\in\follow{\mot{v}}{\cc{q^{|\mot{v}|+1}}}\).
\end{proof}

\begin{lemma}\label{lem-precede}
  Let \(\aut{A}\) be a bireversible Mealy automaton with state
set~\(Q\), \(q\in Q\) a state of~\(\aut{A}\), and \(\mot{u}\) and
\(\mot{v}\) be two elements of~\(Q^+\). If 
  \(\mot{uv}\) is a state of \(\cc{q^{|\mot{uv}|}}\), then \(\mot{u}\)
  is a state of \(\cc{q^{|\mot{u}|}}\) and
  \[\precede{\mot{uv}}{\cc{q^{|\mot{uv}|+1}}}\subseteq \precede{\mot{u}}{\cc{q^{|\mot{u}|+1}}}\:.\]
\end{lemma}

As \(\follow{q^n}{\cc{q^{n+1}}}\subseteq
\follow{q^{n+1}}{\cc{q^{n+2}}}\) by Lemma~\ref{lem-follow}, it is
straightforward to see that the sequence of ratios associated to~\(q\) decreases:
\[\forall n>0,\, \frac{\#\cc{q^{n+2}}}{\#\cc{q^{n+1}}}\leq
\frac{\#\cc{q^{n+1}}}{\#\cc{q^n}}\:,\]
and hence is ultimately constant. We say that \(q\) \emph{has a constant
  ratio\/} if this sequence is in fact constant, and then the unique
value of the sequence of ratios associated to~\(q\) is called the
\emph{ratio of~\(q\)\/}.

It has been proven in~\cite{klimann_ps:3state} that \(q\) induces an
action of infinite order if and only if the sizes of the components
\((\cc{q^n})_{n>0}\) are unbounded, \ie the limit of the sequence of
ratios associated to~\(q\) is greater than~1.

We study now some properties on followers and predecessors in the
components~\(\cc{q^n}\), when~\(q\) has a constant ratio.

The next lemma is an improvement of Lemma~\ref{lem-follow}.

\begin{lemma}\label{lem-followletter}
  Let \(\aut{A}\) be a bireversible Mealy automaton, and \(q\)
be a state of~\(\aut{A}\) of constant ratio.
  Let \(\mot{u}\) and \(\mot{v}\) be two elements of~\(Q^+\) such that
  \(\mot{uv}\) is a state of \(\cc{q^{|\mot{uv}|}}\). We have:
  \[\follow{\mot{uv}}{\cc{q^{|\mot{uv}|+1}}}= \follow{\mot{v}}{\cc{q^{|\mot{v}|+1}}}\:.\]
\end{lemma}

\begin{proof}
The left part is a subset of the right one by Lemma~\ref{lem-follow}
and both sets have the same cardinality, which is the ratio of~\(q\),
by hypothesis.
\end{proof}

In particular, by taking the word \(\mot{v}\) of length~\(1\) in the
previous lemma, we can see that the set of followers of a word~\(\mot{w}\)
in~\(\cc{q^n}\) only depends on its last letter~\(\lettre{-1}{\mot{w}}\).

\begin{lemma}\label{lem-precedeletter}
Let \(\aut{A}\) be a bireversible Mealy automaton of state set~\(Q\),
\(q\) be a state of~\(\aut{A}\) of constant ratio,
and \(n>1\) be an integer. If~\(\mot{u}\in Q^+\) is a suffix of some
state of \(\cc{q^n}\), then the set~\(\precede{\mot{u}}{\cc{q^n}}\) only depends
on \(\lettre{0}{\mot{u}}\), the first letter of~\(\mot{u}\).
\end{lemma}

The next lemma links up the sets of followers and of predecessors in
\(\cc{q^n}\) when \(q\) has a constant ratio.

\begin{lemma}\label{lem-samecardFP}
Let \(\aut{A}\) be a bireversible Mealy automaton of state set~\(Q\),
\(q\) be a state of~\(\aut{A}\) of constant ratio,
and \(n>1\) be an integer. The sets of followers and of predecessors
in~\(\cc{q^n}\) have the same cardinality which is the ratio of~\(q\).
\end{lemma}

\begin{proof}
By Lemmas~\ref{lem-followletter} and~\ref{lem-precedeletter}, we only
have to prove
that \[|\follow{q}{\cc{q^n}}|=|\precede{q}{\cc{q^n}}|\:.\]

Even simpler: notice that if \(\mot{u}\) is a prefix of some state
of~\(\cc{q^n}\), then it is also a prefix of some state
in~\(\cc{q^{n+k}}\) for any \(k>0\), and it has the same followers in
both. Of course, an equivalent property holds for the sets of
predecessors.

So it  is sufficient to
  prove that \[|\follow{q}{\cc{q^2}}|=|\precede{q}{\cc{q^2}}|\:,\]
  which is true because
  \[|\cc{q^2}| = |\cc{q}|\times|\follow{q}{\cc{q^2}}| =
  |\precede{q}{\cc{q^2}}|\times |\cc{q}|\:.\]
\end{proof}

\section{Several equivalences on words}\label{sec-equiv}
\subsection{Minimization and Nerode classes}\label{sec-min}
Let $\aut{A}=(Q,\Sigma,\delta,\rho)$ be a Mealy automaton.

The \emph{Nerode equivalence \(\equiv\) on \(Q\)\/} is the limit of the
sequence of increasingly finer equivalences~$(\equiv_k)$ recursively
defined by:
\begin{align*}
\forall p,q\in Q,\qquad\qquad p\equiv_0 q & \ \Longleftrightarrow
\ \rho_p=\rho_q\:,\\
\forall k\geqslant 0,\ p\equiv_{k+1} q &
\ \Longleftrightarrow\  \bigl(p\equiv_k q\quad \wedge\quad\forall
i\in\Sigma,\ \delta_i(p)\equiv_k\delta_i(q)\bigr)\:.
\end{align*}

Since the set $Q$ is finite, this sequence is ultimately constant.
For every element~$q$ in~$Q$, we denote by~$[q]$  the
class of~$q$ \wrt the Nerode equivalence, called 
the \emph{Nerode class\/}  of
\(q\). Extending to the \(n\)-th power of \(\aut{A}\), we denote 
by \([\mot{u}]\) the Nerode class in \(Q^n\) of
\(\mot{u}\in Q^n\).

Two states of a Mealy automaton belong to the
same Nerode class if and only if they represent
the same element in the generated (semi)group, \ie if and only
if they induce the same action on \(\Sigma^*\). Two
words on \(Q\) of the same length~\(n\) are \emph{equivalent\/} if they
belong to the same Nerode class in \(Q^n\). By extension, any two words on
\(Q\) are \emph{equivalent\/} if they induce the same action.

\medskip

The \emph{minimization\/} of $\aut{A}$ is the Mealy automaton
\(\mz(\aut{A})=(Q/\mathord{\equiv},\Sigma,\tilde{\delta},\tilde{\rho})\),
where for every $(q,i)$ in $Q\times \Sigma$,
$\tilde{\delta}_i([q])=[\delta_i(q)]$ and
$\tilde{\rho}_{[q]}=\rho_q$.
This definition is consistent with the standard minimization of
``deterministic finite automata'' where instead of
considering the mappings $(\rho_q\colon\Sigma\to\Sigma)_q$, the computation
is initiated by the separation between terminal and non-terminal
states.

A Mealy automaton is \emph{minimal\/} if it has the same size as its
minimization.

Two states of two different connected reversible minimal Mealy automata with the
same alphabet induce the same action if and only if the automata are
isomorphic and both states are in correspondance by this
isomorphism. As a direct consequence, if two connected reversible
minimal Mealy automata have different sizes, then any two states of each
of them cannot be equivalent.

As we have seen in Section~\ref{sec-powers}, a state~\(q\) of an
invertible-reversible Mealy 
automaton induces an action of infinite order if and only if the sizes of the
\((\cc{q^n})_{n>0}\) are unbounded. The proof
of~\cite{klimann_ps:3state} can be easily adapted to see that \(q\)
induces an action of infinite order if and only if the sizes of the
\((\mz(\cc{q^n}))_{n>0}\) are unbounded, but you can see it by a direct
argument: if the sizes are bounded, there is an infinite set
\(I\subseteq\N\) such that all the element \((\mz(\cc{q^n}))_{n\in
  I}\) are isomorphic, and in this sequence there exist at least two different
integer \(i\neq j\) such that \(q^i\) and \(q^j\) are represented by
the same state in the minimal automata, and so they induce the same
action; if the sizes are unbounded, the sequence
\((\#\mz(\cc{q^n}))_{n>0}\) has infinitely 
many values, and each value corresponds to a different action for
the corresponding power of \(q\).

Note that the Nerode classes of a connected reversible Mealy automaton
have the same cardinality. The size of the minimization automaton in
this case is the ratio between its size and the cardinality of the
Nerode classes.

\begin{lemma}\label{lem-lastletterNerode}
  Let \(\aut{A}\) be a connected bireversible Mealy automaton,
  \(N\) a Nerode class of a connected component of one of its powers,
  and \(p\) and~\(q\) two elements of~\(\lettre{-1}{N}\). There are as
  many elements of \(N\) with last letter~\(p\) as with last letter~\(q\).
\end{lemma}

\begin{proof}
The proof of this lemma is quite similar to the proof of
Lemma~\ref{lem-precedeletter} considering not words of
length~\(1\) as predecessors, but words on length~\(n-1\), where \(n\)
is the length of the states in~\(N\).
\end{proof}

\subsection{Restricted Nerode classes}
When the considered automaton is not connected, it
can be interesting to consider the restriction of the Nerode class of
an element to its connected component: we denote it by
\(\ccNerode{q}\) et call it the \emph{restricted Nerode class\/}
of~\(q\).

\begin{lemma}\label{lem-samesizerestricted}
The restricted Nerode classes of two elements in the
same connected component of a reversible Mealy automaton have the same
cardinality.
\end{lemma}

Let \(\aut{A}=(Q,\Sigma,\delta,\rho)\) be a bireversible
automaton and \(q\) a state of~\(\aut{A}\) of constant
ratio. As it will be discussed in Section~\ref{sec-main}, the result
of this article is somehow a generalization of a much simpler result proved
in~\cite{Klimann2016}, in the case where all the powers of~\(\aut{A}\)
are connected. The strategy used then is
based on the fact that \(\ccNerode{q^n}q\subseteq \ccNerode{q^{n+1}}\)
when all the powers of \(\aut{A}\) are
connected. However in the more general case we study here, this fact is
false: a priori there is no inclusion 
link between \(\ccNerode{q^n}\) and \(\ccNerode{q^{n+1}}\), because if
\(\mot{u}\) is a state of~\(\ccNerode{q^n}\), then nothing ensures that
\(\mot{u}q\) belongs to~\(\cc{q^{n+1}}\); hence we have
to find a different strategy. For this purpose, we introduce the
\emph{\(q\)-restricted Nerode class\/} of \(q^n\), \ie the
set of states of \(\ccNerode{q^n}\) which 
admit \(q\) as a suffix: \(\Nerendq{n} = \ccNerode{q^n}\cap Q^*q\).

The aim of this section is to study the sequence
\((\Nerendq{n})_{n>0}\).

\begin{lemma}\label{lem-inclusion}
There is an inclusion link in the sequence of \(q\)-restricted Nerode
classes:
\[\forall n>0,\, \Nerendq{n}q\subseteq \Nerendq{n+1}\:.\]
\end{lemma}

\begin{proof}
Let \(\mot{u}\) be an element of \(\Nerendq{n}\): \(\mot{u}\) is a
state of \(\cc{q^n}\) and \(\follow{\mot{u}}{\cc{q^{n+1}}}\) depends
only on \(\lettre{-1}{\mot{u}}=q\) from
Lemma~\ref{lem-followletter},
so in particular \(q\) can follow~\(\mot{u}\) since
\(q\in\follow{q}{\cc{q^{n+1}}}\). Since \(\mot{u}\) and~\(q^n\) induce
the same action, so do \(\mot{u}q\) and~\(q^{n+1}\).
\end{proof}

The next results give more information on the growth of the
\(q\)-restricted Nerode classes of \(q^n\) with respect to~\(n\).

\begin{proposition}\label{prop-boundingrestrNerclasses}
Let \(\aut{A}\) be a bireversible Mealy automaton, and \(q\)
be a state of~\(\aut{A}\) of constant ratio~\(k\). The ratio between
the sizes of~\(\Nerendq{n+1}\) and~\(\Nerendq{n}\) is an integer and:
\[\forall n>0,\,\frac{|\Nerendq{n+1}| }{|\Nerendq{n}|} =
|\lettre{-2}{\Nerendq{n+1}}|\:.\] In particular this ratio cannot be
greater than~\(k\).
\end{proposition}

\begin{proof}
If \(p\in\lettre{-2}{\Nerendq{n+1}}\), then the
sets~\(\Nerendq{n+1}\cap Q^*pq\) and~\(\Nerendq{n+1}\cap Q^*qq\) have
the same cardinality; indeed, an element \(\mot{u}\in Q^*p\) satisfies
\(\mot{u}q\in\Nerendq{n+1}\) if and only if
\(\mot{u}\in\ccNerode{q^n}\) by Lemmas~\ref{lem-inclusion}
and~\ref{lem-followletter} (since by hypothesis~\(p\) can precede~\(q\) in
\(\ccNerode{q^{n+1}}\)). So by Lemma~\ref{lem-lastletterNerode}, 
\(|\ccNerode{q^n}\cap Q^*p|=|\Nerendq{n}|\) and the result follows.

Since \(p\in\lettre{-2}{\Nerendq{n+1}}\) can precede~\(q\)
in~\(\cc{q^{n+1}}\), we obtain the bound \(k=\precede{q}{\cc{q^2}}\)
from Lemmas~\ref{lem-precedeletter} and~\ref{lem-samecardFP}.
\end{proof}

\begin{proposition}\label{prop-stabilizingrestrNerclassesbest}
Let \(\aut{A}\) be a bireversible Mealy automaton, and \(q\)
be a state of~\(\aut{A}\) of constant ratio~\(k\). The sequence
\[\left(\frac{|\Nerendq{n+1}| }{|\Nerendq{n}|} \right)_{n>0}\]
is ultimately increasing to a limit less than or equal to~\(k\).
\end{proposition}

\begin{proof}
Consider the sequence \((\Nerendq{n})_{n>0}\) and particularly the
sequence \((\lettre{0}{\Nerendq{n}})_{n>0}\): from
Lemma~\ref{lem-inclusion}, this sequence increases. 
Denote by~\(Q_1\) its limit: \(Q_1=\lettre{0}{\Nerendq{n}}\) for \(n\)
large enough, say~\(n\geq N\); in particular this set contains~\(q\).

Now, suppose~\(n> N\) and take \(p\in \lettre{-2}{\Nerendq{n+1}}\): clearly
\(\lettre{0}{(\Nerendq{n+1}\cap   Q^*pq)}\) is a subset of
\(Q_{1}=\lettre{0}{\Nerendq{n+1}}=\lettre{0}{\Nerendq{n}}\); since it
has the same cardinality as the set~\(\lettre{0}{(\Nerendq{n+1}\cap
  Q^*qq)}=\lettre{0}{\Nerendq{n}}\) by
Lemma~\ref{lem-lastletterNerode}, it is in fact equal
to~\(Q_{1}\). But \(q\) belongs to~\(Q_{1}\), so it
means that the set~\(\Nerendq{n+1}\cap qQ^*pq\) is not empty,
take~\(\mot{u}\) one of its elements. By
Lemma~\ref{lem-precedeletter}, \(q\) can precede \(\mot{u}\) in
\(\cc{q^{n+2}}\)  and \(q\mot{u}\in\Nerendq{n+2}\). Hence any penultimate letter
in~\(\Nerendq{n+1}\) is also a penultimate letter in~\(\Nerendq{n+2}\):
\[\lettre{-2}{\Nerendq{n+1}}\subseteq \lettre{-2}{\Nerendq{n+2}}\:.\]
The result is now a direct consequence of Proposition~\ref{prop-boundingrestrNerclasses}
\end{proof}

\section{Main result}\label{sec-main}
The following theorem is proved in~\cite{Klimann2016}:
\begin{theorem}\label{thm-kli16}
A semigroup generated by an invertible-reversible Mealy automaton
whose all powers are connected has exponential growth.
\end{theorem}

The main result of this article, Theorem~\ref{thm-main}, is somehow a
generalization of this result because all the elements of a
semigroup generated by an invertible-reversible Mealy automaton whose
all powers are connected have infinite
order~\cite{klimann_ps:3state}. To prove this generalization, we need
to reinforce the hypothesis on the structure of the automaton which is
supposed here to be
bireversible and not only invertible-reversible, but we do not use
anymore the really strong hypothesis on the connected powers. Since it
is (easily) decidable if a Mealy automaton is bireversible, while the
condition on the powers is not known to be decidable or undecidable,
except in very restricted cases, the result here is way more
interesting and powerful, but it is also more tricky to establish. The
question of the existence of elements of infinite order in a
semigroup generated by a bireversible automaton is under study for
several years now and has been solved in quite a few
cases~\cite{Kli13,klimann_ps:3state,GK17}.

Here is the generalized version of the former theorem we prove in this section:

\begin{theorem}\label{thm-main}
A group generated by a bireversible Mealy automaton which contains
an element of infinite order has exponential growth.
\end{theorem}

Note that the result still holds for the generated semigroup and the
proof is easily adaptable.

\begin{proof}
Let \(\aut{A}=(Q,\Sigma,\delta,\rho)\) be a bireversible Mealy
automaton and \(\mot{u}\in (Q\sqcup Q^{-1})^*\) which induces an
action of infinite order. The Mealy automaton
\((\aut{A}\cup\aut{A}^{-1})^{|\mot{u}|}\) is bireversible,
\(\mot{u}\) is one if its state, and it generates a subgroup
of~\(\aut{A}\). So without lost of generality, we can suppose that
\(\mot{u}\) is in fact a state of~\(\aut{A}\). To be consistent with
the rest of the article, let us call it~\(q\).

For any integer \(i>0\), note \(r_i=\frac{\#\cc{q^{i+1}}}{\#\cc{q^i}}\).
The sequence of ratios associated to~\(q\) is of the form
\((r_1,r_2,\ldots,r_j,r_{j+1}=r_j,\ldots)\), where \(r_i\geq r_{i+1}\) for any
\(i\geq 1\) and \(r_i=r_j\) for any \(i\geq j\). Now, consider the 
component \(\cc{q^j}\) as a Mealy automaton and \(\mot{q}=q^j\) as its
state: the state \(\mot{q}\) induces an action of infinite 
order and the sequence of ratios associated to
\(\mot{q}\) is of the form \((r_j^j,r_j^j,\ldots)\): \(\mot{q}\) has
constant ratio. Moreover, \(\cc{\mot{q}}\) generates a subgroup of
the group \(\pres{\aut{A}}\), so if we prove that
\(\pres{\cc{\mot{q}}}\) has exponential growth, so has
\(\pres{\aut{A}}\). So, without loss of generality we
can suppose that \(q\) has a constant ratio, say~\(k\).

It is quite immediate to obtain the following inequalities:
\begin{equation}\label{eq-main}
\forall n>0,\,\frac{|\cc{q^n}|}{|Q|\times |\Nerendq{n}|}\leq \frac{|\cc{q^n}|}{|\ccNerode{q^n}|}\leq \frac{|\cc{q^n}|}{|\Nerendq{n}|}\:.
\end{equation}
Indeed, the right part is a consequence of the fact that
\(\Nerendq{n}\subseteq\ccNerode{q^n}\), and the left part of
Lemma~\ref{lem-lastletterNerode} and the fact
that \(\ccNerode{q^n}=\cup_{p\in Q}(\ccNerode{q^n}\cap Q^*p)\).

The central part in~\eqref{eq-main} is in fact the size of the
minimization of \(\cc{q^n}\), and the cardinality of~\(\cc{q^n}\) is
equal to~\(|Q|\times k^{n-1}\) since \(q\) has constant ratio~\(k\),
so~\eqref{eq-main} can be re-written as:
\begin{equation}\label{eq-main2}
\forall n>0,\,\frac{k^{n-1}}{|\Nerendq{n}|}\leq \#\mz(\cc{q^n})\leq
\frac{|Q|\times k^{n-1}}{|\Nerendq{n}|}\:.
\end{equation}

Let us prove that for \(n\) large enough, \(|\Nerendq{n+1}|<k\times
|\Nerendq{n}|\). From Proposition~\ref{prop-boundingrestrNerclasses}
we know that the corresponding non strict inequality is
satisfied. Now, from
Propositions~\ref{prop-boundingrestrNerclasses}
and~\ref{prop-stabilizingrestrNerclassesbest} we know that
for~\(N\) large enough, if the equality holds at rank~\(N\), it also holds
at any rank greater than~\(N\): for any \(n\geq N\),
\(|\Nerendq{n+1}|=k\times|\Nerendq{n}|\). This means that for
\(n\) large enough, the value of the \(n\)-th term of the sequence
\((|\Nerendq{n}|)_{n>0}\) is \(ck^{n-1}\), where \(c\) does not depend
on~\(n\). Hence by the right part of Equation~\eqref{eq-main2}, the
minimizations of the connected powers of~\(q\) have bounded size,
which implies that~\(q\) has finite order as seen in
Section~\ref{sec-min}, and this is in contradiction with the hypotheses.

Denote by \(\ell\) the limit of the sequence
\[\left(\frac{|\Nerendq{n+1}| }{|\Nerendq{n}|} \right)_{n>0}\]
(it exists from
Proposition~\ref{prop-stabilizingrestrNerclassesbest}): \(\ell>1\)
because \(q\) has infinite order and \(\ell<k\) from the above
paragraph. For \(n\) large enough, there exists a constant~\(c\) such
that \(|\Nerendq{n}|=\frac{\ell^{n}}{kc}\).

So Equation~\eqref{eq-main2} becomes:
\begin{equation}\label{eq-main3}
\forall n>N,\,c\left(\frac{k}{\ell}\right)^n\leq \#\mz(\cc{q^n})\leq
c|Q|\left(\frac{k}{\ell}\right)^{n}\:.
\end{equation}

Since \(\ell<k\), there exists \(\alpha\) such that
\(\left(\frac{k}{l}\right)^{\alpha}>|Q|\). Let us denote
\(\mot{u}=q^{\alpha}\) and \(K=\left(\frac{k}{\ell}\right)^{\alpha}\),
we have that for \(n\) large enough:
\begin{equation}\label{eq-main4}
c\cdot K^n\leq \#\mz(\cc{\mot{u}^n})\leq
c\cdot|Q|\cdot K^{n}<c\cdot K^{n+1}\leq \#\mz(\cc{\mot{u}^{n+1}})\leq
c\cdot|Q|\cdot K^{n+1}\:.
\end{equation}

Consequently, the minimizations of the components~\(\cc{\mot{u}^n}\)
are pairwise not isomorphic, for~\(n\) large enough, because they do not have
the same size. So their states induce different elements of the
group~\(\aut{A}\). Hence the sets
\[I_n=\{\rho_{\mot{v}}\mid
\mot{v}\text{ is a state of }\mz(\cc{\mot{u}^n})\}\]
are pairwise disjoint. By Equation~\eqref{eq-main4}, the growth of
the sequence \((I_n)_{n>0}\) is exponential, and so is the growth of~\(\aut{A}\).
\end{proof}

By combining Theorem~\ref{thm-main} with the fact that connected
bireversible Mealy automata of prime size cannot generate infinite
Burnside groups~\cite{GK17}, we have:
\begin{corollary}
Any infinite group generated by a bireversible connected Mealy automaton of
prime size has exponential growth.
\end{corollary}

Another consequence of Theorem~\ref{thm-main} concerns virtually
nilpotent groups. This class is important in the classification of
groups and contains in particular all the abelian groups. It is known
that any infinite virtually nilpotent group contains an element of
infinite order~\cite[Proposition~10.48]{DK06} and has polynomial
growth~\cite{wolf1968}. This leads to

\begin{corollary}\label{cor-virtnil}
No infinite virtually nilpotent group can be generated by a bireversible Mealy automaton.
\end{corollary}

\subsection*{Acknowledgement} The author is very grateful to Dmytro
Savchuk for suggesting Corollary~\ref{cor-virtnil}.

\bibliography{exponentialgrowth_hal}

\end{document}